\documentclass[12pt]{article}

\usepackage{times}
\usepackage{fullpage}
\usepackage{amsfonts,amssymb,amsmath}
\usepackage{latexsym}
\usepackage{url}
\usepackage[pdftex]{color}
\usepackage[breaklinks]{hyperref}

\def\01{\{0,1\}}

\newcommand{\Tr}{\mathrm{Tr}}

\newtheorem{theorem}{Theorem}

\newcommand{\thmref}[1]{\hyperref[#1]{{Theorem~\ref*{#1}}}}
\newcommand{\lemref}[1]{\hyperref[#1]{{Lemma~\ref*{#1}}}}
\newcommand{\corref}[1]{\hyperref[#1]{{Corollary~\ref*{#1}}}}
\newcommand{\eqnref}[1]{\hyperref[#1]{{Equation~(\ref*{#1})}}}
\newcommand{\claimref}[1]{\hyperref[#1]{{Claim~\ref*{#1}}}}
\newcommand{\remarkref}[1]{\hyperref[#1]{{Remark~\ref*{#1}}}}
\newcommand{\propref}[1]{\hyperref[#1]{{Proposition~\ref*{#1}}}}
\newcommand{\factref}[1]{\hyperref[#1]{{Fact~\ref*{#1}}}}
\newcommand{\defref}[1]{\hyperref[#1]{{Definition~\ref*{#1}}}}
\newcommand{\exampleref}[1]{\hyperref[#1]{{Example~\ref*{#1}}}}
\newcommand{\hypref}[1]{\hyperref[#1]{{Hypothesis~\ref*{#1}}}}
\newcommand{\secref}[1]{\hyperref[#1]{{Section~\ref*{#1}}}}
\newcommand{\chapref}[1]{\hyperref[#1]{{Chapter~\ref*{#1}}}}
\newcommand{\apref}[1]{\hyperref[#1]{{Appendix~\ref*{#1}}}}

\newenvironment{proof}[1][Proof: ]
{\noindent {\bf #1}}
{{\hfill $\Box$}\\
\smallskip}

\newcommand{\veczero}{\mathbf{0}}
\newcommand{\A}{\mathbf{A}}
\newcommand{\B}{\mathbf{B}}
\newcommand{\Lov}{Lov{\'a}sz }
\newcommand{\ADV}{\mathrm{ADV}}

\begin{document}

\title{Product theorems via semidefinite programming}
\author{Troy Lee \\ Department of Computer Science \\ Rutgers University
\thanks{Supported by a NSF Mathematical Sciences Postdoctoral Fellowship.
Email: troyjlee@gmail.com}
\and
Rajat Mittal \\ Department of Computer Science \\ Rutgers University 
\thanks{Supported by NSF Grant 0523866.  Email: ramittal@cs.rutgers.edu}}
\date{}
\maketitle


\begin{abstract}
The tendency of semidefinite programs to compose perfectly under product has been exploited 
many times in complexity theory: for example, by Lov{\'a}sz to determine the Shannon capacity of
the pentagon; to show a direct sum theorem for non-deterministic communication complexity and 
direct product theorems for discrepancy; and in interactive proof systems to show parallel 
repetition theorems for restricted classes of games.

Despite all these examples of product theorems---some going back nearly thirty years---it was 
only recently that Mittal and Szegedy began to develop a general theory to explain when and 
why semidefinite programs behave perfectly under product.  This theory captured many examples in 
the literature, but there were also some notable exceptions which it could not explain---namely, an 
early parallel repetition result of Feige and Lov{\'a}sz, and a direct product theorem for the 
discrepancy method of communication complexity by Lee, Shraibman, and \v{S}palek.  

We extend the theory of Mittal and Szegedy to explain these cases as well.  Indeed, to the
best of our knowledge, our theory captures all examples of semidefinite product theorems 
in the literature.  
\end{abstract}

\section{Introduction}
A prevalent theme in complexity theory is what we might roughly call product theorems.  These
results look at how the resources to accomplish several independent tasks scale with the resources
needed to accomplish the tasks individually.  Let us look at a few examples of such questions:

\paragraph{Shannon Capacity}
If a graph $G$ has an independent set of size $\alpha$, how large an independent set 
can the product graph $G \times G$ have?  How does $\alpha$ compare with amortized independent
set size $\lim_{k \rightarrow \infty} \alpha(G^k)^{1/k}$?  This last quantity, known as the 
Shannon capacity, gives the effective alphabet size of a graph where vertices are labeled by 
letters and edges represent letters which can be confused if adjacent.     

\paragraph{Hardness Amplification}
Product theorems naturally arise in the context of hardness amplification.  If it is hard to evaluate a 
function $f(x)$, then an obvious approach to create a harder function is to evaluate two independent
copies $f'(x,y)=(f(x),f(y))$ of $f$.  There are different ways that $f'$ can be harder than $f$---a direct
sum theorem aims to show that evaluation of $f'$ requires twice as many resources as needed to
evaluate $f$; direct product theorems aim to show that the error probability to compute $f'$ 
is larger than that of $f$, given the same amount of resources.    

\paragraph{Soundness Amplification}
Very related to hardness amplification is what we might call soundness amplification.  This arises
in the context of interactive proofs where one wants to reduce the error probability of a protocol,
by running several checks in parallel.  The celebrated parallel repetition theorem shows that 
the soundness of multiple prover interactive proof systems can be boosted in this manner \cite{Raz98}.  

These examples illustrate that many important problems in complexity theory deal with product
theorems.  One successful approach to these types of questions has been through 
semidefinite programming.  In this approach, if one wants to know how some quantity 
$\sigma(G)$ behaves under product, one first looks at a semidefinite approximation 
$\bar \sigma(G)$ of $\sigma(G)$.  One then hopes to show that $\bar \sigma(G)$ provides a 
good approximation to $\sigma(G)$, and that $\bar \sigma(G \times G)=\bar \sigma(G) \bar \sigma(G)$.  
In this way one obtains that the original quantity must approximately product as well. 

Let us see how this approach has been used on some of the above questions.

\paragraph{Shannon Capacity} Perhaps the first application of this technique was to the 
Shannon capacity of a graph.  \Lov developed a semidefinite quantity, the \Lov theta function 
$\vartheta(G)$, showed that it
was a bound on the independence number of a graph, and that 
$\vartheta(G \times G)= \vartheta(G)^2$.  In this way he determined the Shannon capacity of 
the pentagon, resolving a long standing open problem \cite{Lov79}.

\paragraph{Hardness Amplification} Karchmer, Kushilevitz, and Nisan \cite{KKN95} notice that another
program introduced by Lov\'asz \cite{Lov75}, the fractional cover number, can be used to characterize 
non-deterministic communication complexity, up to small factors.  As this program also 
perfectly products, they obtain a direct sum theorem for non-deterministic 
communication complexity.  

As another example, Linial and Shraibman \cite{LS06} show that a semidefinite programming quantity 
$\gamma_2^{\infty}$ characterizes the discrepancy method of communication complexity, up to
constant factors.  Lee, Shraibman and \v{S}palek \cite{LSS08} then use this result, together with the 
fact that $\gamma_2^{\infty}$ perfectly products, to show a direct product theorem for 
discrepancy, resolving an open problem of Shaltiel \cite{Sha03}.

\paragraph{Soundness Amplification}  Although the parallel repetition theorem was eventually
proven by other means \cite{Raz98,Hol07}, one of the first positive results did use 
semidefinite programming.  Feige and \Lov \cite{FL92} show that the acceptance probability 
$\omega(G)$ of a two-prover interactive 
proof on input $x$ can be represented as an integer program.  They then study a 
semidefinite relaxation of this program, and use this to show that if $\omega(G) < 1$ then
$\sup_{k \rightarrow \infty} \omega(G^k)^{1/k} < 1$, for a certain class of games $G$.  
More recently, Cleve et al. \cite{CSUU07} look at two-prover games where the provers share entanglement, and show that the value of a special kind of such a game known as an XOR game can 
be exactly represented by a semidefinite program.  As this program perfectly products, they obtain a 
perfect parallel repetition theorem for this game.

We hope this selection of examples shows the usefulness of the semidefinite programming approach
to product theorems.  Until recently, however, this approach remained an ad hoc collection of 
examples without a theory to explain when and why semidefinite programs perfectly product.
Mittal and Szegedy \cite{MS07} began to address this lacuna by giving a general sufficient condition for 
a semidefinite program to obey a product rule.  This condition captures many examples in the 
literature, notably the \Lov theta function \cite{Lov79}, and the parallel repetition for XOR games with entangled provers \cite{CSUU07}.

Other examples cited above, however, do not fit into the Mittal and Szegedy framework: namely, 
the product theorem of Feige and \Lov \cite{FL92} and that for discrepancy \cite{LSS08}.  
We extend the condition of Mittal and Szegedy to capture these cases as well.  Indeed, in 
our (admittedly imperfect) search of the literature, we have not found a semidefinite product 
theorem which does not fit into our framework.

\section{Preliminaries}
We begin with some notational conventions and basic definitions which will be useful.  In general, 
lower case letters like $v$ will denote column vectors, and upper case letters like $A$ will
denote matrices.  Vectors and matrices will be over the real numbers.  The notation $v^T$ or $A^T$ 
will denote the transpose of a vector or matrix.  We will say $A \succeq 0$ if $A$ is positive semidefinite,
i.e.\ if $A$ is symmetric and $v^T A v \ge 0$ for all vectors $v$.  

We will use several kinds of matrix products.  We write $AB$ for the normal matrix product.  
For two matrices $A,B$ of the same dimensions, $A \circ B$ denotes the matrix formed by their 
entrywise product.  That is, $(A \circ B)[x,y]=A[x,y] B[x,y]$.  We will use $A \bullet B$ for the entrywise 
sum of $A \circ B$.  Equivalently, $A \bullet B =\Tr(AB^T)$.  We will use the notation
$v \ge w$ to indicate that the vector $v$ is entrywise greater than or equal to the vector $w$.

In applications we often face the situation where we would like to use the framework of semidefinite 
programming, which requires symmetric matrices, but the problem at hand is represented by matrices
which are not symmetric, or possibly not even square.  Fortunately, this can often be handled by 
a simple trick.  This trick is so useful that we will give it its own notation.  For an arbitrary real matrix $A$, 
we define
$$
\widehat A=\begin{bmatrix}
0 & A \\
A^T & 0
\end{bmatrix}
$$
We will refer to this as the {\em bipartite version} of $A$, as such a matrix corresponds to the adjacency
matrix of a (weighted) bipartite graph.  In many respects $\widehat A$ behaves similarly to $A$, but
has the advantages of being symmetric and square.  

More generally, we will refer to a matrix $M$ which can be written as
$$
M=\begin{bmatrix}
0 & A \\
B & 0
\end{bmatrix}
$$
as {\em block anti-diagonal} and a matrix $M$ which can be written
$$
M=\begin{bmatrix}
D_1 & 0 \\
0 & D_2
\end{bmatrix}
$$
as {\em block diagonal}.

One subtlety that arises in working with the bipartite version $\widehat A$ instead of $A$
itself is in defining the product of instances.  Mathematically, it is most convenient to work with
the normal tensor product 
$$
\widehat A \otimes \widehat A = \begin{bmatrix}
0&0&0& A \otimes A \\
0&0&A \otimes A^T &0 \\
0& A^T \otimes A &0&0 \\
A^T \otimes A^T &0&0&0
\end{bmatrix}
$$
Whereas what naturally arises in the product of problems is instead the ``bipartite tensor'' product 
of $A$:
$$
\widehat{A \otimes A}=\begin{bmatrix}
0 & A \otimes A \\
A^T \otimes A^T & 0
\end{bmatrix}
$$

Kempe, Regev, and Toner \cite{KRT07} observe, however, that a product theorem for the tensor 
product implies a product theorem for the bipartite tensor product.  This essentially follows because 
$\widehat{A \otimes A}$ is a submatrix of $\widehat A \otimes \widehat A$, and so positive 
semidefiniteness of the latter implies positive semidefiniteness of the former.  See \cite{KRT07} for
full details.  

\section{Product rule with non-negativity constraints}
In this section we prove our main theorem extending the product theorem of Mittal and 
Szegedy \cite{MS07} to handle non-negativity constraints.  As our work builds on the framework
developed by Mittal and Szegedy, let us first explain their results.  

Mittal and Szegedy consider a general affine semidefinite program $\pi=(J,\A,b)$.  
Here $\A=(A_1, \ldots, A_m)$ is a vector of matrices, and we extend the notation $\bullet$ such that 
$\A \bullet X=(A_1\bullet X, A_2 \bullet X, \ldots, A_m \bullet X)$.  The value of $\pi$ is given as
\begin{align*}
\alpha(\pi) =&\max_X J \bullet X \mbox{ such that} \\
          & \A \bullet X =b \\
          & X \succeq 0.
\end{align*}
We take this as the primal formulation of $\pi$.  Part of what makes semidefinite programming
so useful for proving product theorems is that we can also consider the dual formulation of $\pi$.  
Dualizing in the straightforward way gives:
\begin{align*}
\alpha^*(\pi)=& \min_y y^T b \\
     & y^T \A -J \succeq 0
\end{align*}
A necessary pre-condition for the semidefinite programming approach to proving product theorems 
is that so-called strong duality holds.  That is, that
$\alpha(\pi)=\alpha^*(\pi)$, the optimal primal and dual values agree.  We will assume this
throughout our discussion.  For more information about strong duality and sufficient conditions for it 
to hold, see \cite{BV06}. 

We define the product of programs as follows: for $\pi_1=(J_1, \A_1, b_1)$ and 
$\pi_2=(J_2, \A_2, b_2)$ we define $\pi_1 \times \pi_2=(J_1 \otimes J_2, \A_1 \otimes \A_2,
b_1 \otimes b_2)$.  If $\A_1$ is a tuple of $m_1$ matrices and $\A_2$ is a tuple of $m_2$ matrices, 
then the tensor product $\A_1 \otimes \A_2$ is a tuple of $m_1 m_2$ matrices consisting of all the 
tensor products $\A_1[i] \otimes \A_2[j]$.  

It is straightforward to see that $\alpha(\pi_1 \times \pi_2) \ge \alpha(\pi_1) \alpha(\pi_2)$.  Namely,
if $X_1$ realizes $\alpha(\pi_1)$ and $X_2$ realizes $\alpha(\pi_2)$, then $X_1 \otimes X_2$ will
be a feasible solution to $\pi_1 \times \pi_2$ with value $\alpha(\pi_1)\alpha(\pi_2)$.  This is because
$X_1 \otimes X_2$ is positive semidefinite, $(\A_1 \otimes \A_2) \bullet (X_1 \otimes X_2)=
(\A_1 \bullet X_1) \otimes (\A_2 \bullet X_2)=b_1 \otimes b_2$, and 
$(J_1 \otimes J_2) \bullet (X_1 \otimes X_2)=(J_1 \bullet X_1) \otimes (J_2 \bullet X_2) = 
\alpha(\pi_1) \alpha(\pi_2)$.

Mittal and Szegedy show the following theorem giving sufficient conditions for the reverse 
inequality $\alpha(\pi_1 \times \pi_2) \le \alpha(\pi_1) \alpha(\pi_2)$.  

\begin{theorem}[Mittal and Szegedy \cite{MS07}]
\label{MS}
Let $\pi_1=(J_1, \A_1, b_1), \pi_2=(J_2, \A_2, b_2)$ be two affine semidefinite programs 
for which strong duality holds.  
Then $\alpha(\pi_1 \times \pi_2) \le \alpha(\pi_1) \alpha(\pi_2)$ if either of the following two 
conditions hold:
\begin{enumerate}
  \item $J_1, J_2 \succeq 0$.
  \item (Bipartiteness) There is a partition of rows and columns into two sets such that with respect to this partition, 
  $J_i$ is block anti-diagonal, and all matrices in $\A_i$ are block diagonal, for $i \in \{1,2\}$.    
\end{enumerate}
\end{theorem}
 
We extend item~(2) of this theorem to also handle non-negativity constraints.  This
is a class of constraints which seems to arise often in practice, and allows us to capture 
cases in the literature that the original work of Mittal and Szegedy does not.  
More precisely, we consider programs of the following form:
\begin{align*}
\alpha(\pi) =&\max_X J \bullet X \mbox{ such that} \\
          & \A \bullet X =b \\
          & \B \bullet X \ge \veczero \\
          & X \succeq 0
\end{align*}
Here both $\A$ and $\B$ are vectors of matrices, and $\veczero$ denotes the all 0 vector.  

We should point out a subtlety here.  A program of this form can be equivalently written as an 
affine program by suitably extending $X$ and modifying $\A$ accordingly to enforce the 
$\B \bullet X \ge \veczero$ constraints through the $X \succeq 0$ condition.
The catch is that two equivalent programs do not necessarily lead to equivalent product instances.    
We explicitly separate out the non-negativity constraints here so that we can define the product as 
follows: for two programs, 
$\pi_1=(J_1, \A_1, b_1, \B_1)$ and $\pi_2=(J_2, \A_2, b_2, \B_2)$ we say
$$
\pi_1 \times \pi_2=(J_1 \otimes J_2, \A_1 \otimes \A_2, b_1 \otimes b_2, \B_1 \otimes \B_2).
$$
Notice that the equality constraints and non-negativity constraints do not interact in the 
product, which is usually the intended meaning of the product of instances.  

It is again straightforward to see that $\alpha(\pi_1 \times \pi_2) \ge \alpha(\pi_1)\alpha(\pi_2)$, 
thus we focus on the reverse inequality.  We extend Condition (2) of \thmref{MS} to the case of 
programs with non-negativity constraints.  As we will see in \secref{apps}, this theorem captures 
the product theorems of Feige-\Lov \cite{FL92} and discrepancy \cite{LSS08}.

\begin{theorem}
\label{main}
Let $\pi_1=(J_1, \A_1, b_1,\B_1)$ and $\pi_2=(J_2, \A_2, b_2, \B_2)$ be two semidefinite programs
for which strong duality holds.  Suppose the following two conditions hold:
\begin{enumerate}
  \item (Bipartiteness) There is a partition of rows and columns into two sets 
such that, with respect to this partition, $J_i$ and all the matrices of $\B_i$ are block anti-diagonal, 
and all the matrices of $\A_i$ are block diagonal, for $i \in \{1,2\}$.
  \item There are non-negative vectors $u_1, u_2$ such that $J_1=u_1^T \B_1$ and $J_2=u_2^T \B_2$.  
\end{enumerate}
Then $\alpha(\pi_1 \times \pi_2) \le \alpha(\pi_1) \alpha(\pi_2)$.
\end{theorem}

\begin{proof}
To prove the theorem it will be useful to consider the dual formulations of $\pi_1$ and $\pi_2$.
Dualizing in the standard fashion, we find
\begin{align*}
\alpha(\pi_1)=& \min_{y_1} \ y_1^T b_1 \mbox{ such that} \\
                     & y_1^T \A_1 - (z_1^T \B_1 + J_1) \succeq 0 \\
                     & z_1 \ge 0
\end{align*}
and similarly for $\pi_2$.  Fix $y_1, z_1$ to be vectors which realizes this optimum for $\pi_1$ and 
similarly $y_2, z_2$ for $\pi_2$.  The key observation of the proof is that if we can also show that 
\begin{equation}
y_1^T \A_1 + (z_1^T \B_1 + J_1) \succeq 0 \mbox{ and } y_2^T \A_2 + (z_2^T \B_2 + J_2) \succeq 0
\label{plus}
\end{equation}
then we will be done.  Let us for the moment assume \eqnref{plus} and
see why this is the case.  

If \eqnref{plus} holds, then we also have
\begin{align*}
\left(y_1^T \A_1 - (z_1^T \B_1+ J_1) \right)\otimes \left(y_2^T \A_2+(z_2^T \B_2 + J_2)\right) \succeq 0 \\
\left(y_1^T \A_1 + (z_1^T \B_1 + J_1)\right)\otimes \left(y_2^T \A_2 - (z_2^T \B_2 +J_2)\right) \succeq 0
\end{align*}
Averaging these equations, we find
$$
(y_1 \otimes y_2)^T (\A_1 \otimes \A_2) -  \left((z_1^T \B_1 + J_1)\otimes (z_2^T \B_2 + J_2)\right) 
\succeq 0.
$$
Let us work on the second term.  We have
\begin{align*}
(z_1^T \B_1 + J_1)\otimes (z_2^T \B_2 + J_2)&=(z_1 \otimes z_2)^T (\B_1 \otimes \B_2) 
+ z_1^T \B_1 \otimes J_2 + J_1 \otimes z_2^T \B_2 + J_1 \otimes J_2 \\
&=(z_1 \otimes z_2)^T (\B_1 \otimes \B_2)+ (z_1 \otimes u_2)^T \B_1 \otimes \B_2 \\
&+ (u_1 \otimes z_2)^T \B_1 \otimes \B_2 + J_1 \otimes J_2.
\end{align*}

Thus if we let $v=z_1 \otimes z_2 + z_1 \otimes u_2+u_1 \otimes z_2$ we see that $v \ge 0$
as all of $z_1, z_2, u_1, u_2$ are, and also
$$
(y_1 \otimes y_2)^T \otimes (\A_1 \otimes \A_2) - (v^T (\B_1 \otimes \B_2) +J_1 \otimes J_2) \succeq 0.
$$
Hence $(y_1 \otimes y_2, v)$ form a feasible solution to the dual formulation of 
$\pi_1 \times \pi_2$ with value $(y_1 \otimes y_2)(b_1 \otimes b_2)=\alpha(\pi_1) \alpha(\pi_2)$.  

It now remains to show that \eqnref{plus} follows from the condition of the theorem.  
Given $y \A - (z^T\B +J) \succeq 0$ and the bipartiteness condition of the theorem, we will 
show that $y \A + (z^T\B +J) \succeq 0$.  This argument can then be applied to both $\pi_1$ and 
$\pi_2$.  

We have that $y^T \A$ is block diagonal and $z^T \B + J$ is block anti-diagonal 
with respect to the same partition.  Hence for any vector $x^T=\begin{bmatrix} x_1 & x_2\end{bmatrix}$, 
we have 
$$
\begin{bmatrix}
x_1 & x_2
\end{bmatrix}
\left(y^T \A - (z^T \B + J)\right)
\begin{bmatrix}
x_1 \\
x_2
\end{bmatrix} 
 =
 \begin{bmatrix}
 x_1 & -x_2
 \end{bmatrix} 
 \left(y^T \A + (z^T \B + J) \right) 
 \begin{bmatrix}
 x_1 \\
 -x_2
 \end{bmatrix}
$$
Thus the positive semidefiniteness of $y \A + (z^T\B +J)$ follows from that of $y \A - (z^T\B +J)$.
\end{proof}

One may find the condition that $J$ lies in the positive span of $\B$ in the statement of
\thmref{main} somewhat unnatural.  If we remove this condition, however, a simple counterexample shows that the theorem no longer holds.  Consider the program
\begin{align*}
\alpha(\pi) =&\max_X \ \begin{bmatrix}0 & -1 \\ -1 & 0 \end{bmatrix} \bullet X \\
& \mbox{ such that } I \bullet X =1,
 \begin{bmatrix}0 & 1 \\ 0 & 0\end{bmatrix} \bullet X \ge 0,
\begin{bmatrix}0 & 0 \\ 1 & 0\end{bmatrix} \bullet X \ge 0,
X \succeq 0.
\end{align*}
Here $I$ stands for the $2$-by-$2$ identity matrix.  This program satisfies the bipartiteness 
condition of \thmref{main}, but $J$ does not lie in the positive
span of the matrices of $\B$.  It is easy to see that the value of this program is zero.  The program
$\pi \times \pi$, however, has positive value as $J \otimes J$ does not have any negative entries
but is the matrix with ones on the main anti-diagonal.   

\section{Applications}
\label{apps}
Two notable examples of semidefinite programming based product theorems in the literature are not
captured by \thmref{MS}.  Namely, a recent direct product theorem for the discrepancy method of
communication complexity, and an early semidefinite programming based parallel repetition
result of Feige and Lov\'asz.  As we now describe in detail, these product theorems can be explained
by \thmref{main}.

\subsection{Discrepancy}
Communication complexity is an ideal model to study direct sum and direct product theorems as it 
is simple enough that one can often hope to attain tight results, yet powerful enough that such
theorems are non-trivial and have applications to reasonably powerful models of computation.  
See \cite{KN97} for more details on communication complexity and its applications.

Shaltiel \cite{Sha03} began a systematic study of when we can expect direct product theorems 
to hold, and in particular looked at this question in the model of communication complexity for 
exactly these reasons.  He showed a general counterexample where a direct product theorem does 
not hold, yet also proved a direct product for communication complexity lower bounds shown by 
a particular method---the discrepancy method under the uniform distribution.  Shaltiel does not 
explicitly use semidefinite programming techniques, but proceeds by relating discrepancy under 
the uniform distribution to the spectral norm, which can be cast as a semidefinite program.

This result was recently generalized and strengthened by Lee, Shraibman, and \v{S}palek 
\cite{LSS08} who show an essentially optimal direct product theorem for discrepancy under 
arbitrary distributions.  This result follows the general plan for showing product theorems via
semidefinite programming: they use a result of Linial and Shraibman \cite{LS06} that  a semidefinite 
programming quantity $\gamma_2^{\infty}(M)$ characterizes the discrepancy of the communication
matrix $M$ up to a constant factor, and then show that $\gamma_2^{\infty}(M)$ perfectly products.  
The semidefinite programming formulation of $\gamma_2^{\infty}(M)$ is not affine but involves
non-negativity constraints, and so does not fall into the original framework of Mittal and Szegedy.  

Let us now look at the semidefinite program describing $\gamma_2^{\infty}$:
\begin{align*}
\gamma_2^\infty(M)= & \max_X \ \widehat M \bullet X \mbox{ such that} \\
& X \bullet I = 1 \\
& X \bullet E_{ij}=0 \mbox{ for all } i \ne j \le m, i \ne j \ge m \\
& X \bullet (\widehat M \circ E_{ij}) \ge 0 \mbox{ for all } i \le m, j \ge m, \mbox{ and } i \ge m, j \le m \\
& X \succeq 0.
\end{align*}
Here $E_{i,j}$ is the 0/1 matrix with exactly one entry equal to $1$ in coordinate $(i,j)$.  
In this case, $\A$ is formed from the matrices $I$ and $E_{ij}$ for $i \ne j \le m$ and $i \ne j \ge m$.  
These matrices are all block diagonal with respect to the natural partition of $\widehat M$.  
Further, the objective matrix $\widehat M$ and matrices of $\B$ are all block anti-diagonal with 
respect to this partition.  Finally, we can express $\widehat M=u^T \B$ by simply taking $u$ to be the 
all 1 vector.  

\subsection{Feige-Lov\'asz}
In a seminal paper, Babai, Fortnow, and Lund \cite{BFL91} show that all of non-deterministic exponential time can be captured by interactive proof systems with two-provers and polynomially many rounds.  The attempt to characterize the power of two-prover systems with just one round sparked 
interest in a parallel repetition theorem---the question of whether the soundness of a two-prover 
system can be amplified by running several checks in parallel.  Feige and \Lov \cite{FL92} ended up 
showing that two-prover one-round systems capture NEXP by other means, and a proof of a 
parallel repetition theorem turned out to be the more difficult question \cite{Raz98}.
In the same paper, however, Feige and \Lov also take up the study of parallel repetition theorems
and show an early positive result in this direction.  

In a two-prover one-round game, the Verifier 
is trying to check if some input $x$ is in the language $L$.  The Verifier chooses questions 
$s \in S, t \in T$ with some probability $P(s,t)$ and then sends question $s$ to prover Alice, and 
question $t$ to prover Bob.  
Alice sends back an answer $u \in U$ and Bob replies $w \in W$, and then the Verifier answers 
according to some Boolean predicate $V(s,t,u,w)$.  We call this a game $G(V,P)$, and write the
acceptance probability of the Verifier as $\omega(G)$.  
In much the same spirit as the result of \Lov on the Shannon capacity of a graph, Feige and 
\Lov show that if the value of a game $\omega(G) <1$ then also $\sup_k \omega(G^k)^{1/k} < 1$, 
for a certain class of games known as unique games.

The proof of this result proceeds in the usual way: Feige and \Lov first show that $\omega(G)$ 
can be represented as a quadratic program.  They then relax this quadratic program in the natural
way to obtain a semidefinite program with value $\sigma(G) \ge \omega(G)$.  Here the proof faces
an extra complication as $\sigma(G)$ does not perfectly product either.  Thus another round of 
relaxation is done, throwing out some constraints to obtain a program with value 
$\bar \sigma(G) \ge \sigma(G)$ which does perfectly product.  Part of our motivation for proving 
\thmref{main} was to uncover the ``magic'' of this second round of relaxation, and explain why 
Feige and \Lov remove the constraints they do in order to obtain something which perfectly products.

Although the parallel repetition theorem was eventually proven by different means 
\cite{Raz98,Hol07}, the semidefinite programming approach has recently seen renewed
interest for showing tighter parallel repetition theorems for restricted classes of games and where 
the provers share entanglement \cite{CSUU07,KRT07}.

\subsubsection{The relaxed program}
As mentioned above, Feige and \Lov first write $\omega(G)$ as an integer program, and 
then relax this to a semidefinite program with value $\sigma(G) \ge \omega(G)$.  We now describe
this program.  The objective matrix $C$ is a $|S| \times |U|$-by-$|T| \times |W|$ matrix where the rows are labeled by
pairs $(s,u)$ of possible question and answer pairs with Alice and similarly the columns are labeled
by $(t,w)$ possible dialogue with Bob.  The objective matrix for a game $G=(V,P)$ is given by
$C[(s,u),(t,w)]=P(s,t) V(s,t,u,w)$.  We also define an auxiliary matrices $B_{st}$ of dimensions the
same as $\widehat C$, where $B_{st}[(s',u), (t',w)]=1$ if $s=s'$ and $t=t'$ and is zero otherwise.  

With these notations in place, we can define the program:
\begin{align}
\sigma(G)=&\max_{X} \frac{1}{2} \widehat C \bullet X \mbox{ such that} \\
& X \bullet B_{st} =1 \mbox{ for all } s,t \in S \cup T \\
& X \ge 0 \\
& X \succeq 0
\end{align}
We see that we cannot apply \thmref{main} here as we have global non-negativity constraints 
(not confined to the off-diagonal blocks) and global equality constraints (not confined to the 
diagonal blocks).  Indeed, Feige and \Lov remark that this program does not perfectly product.  

Feige and \Lov then consider a further relaxation with value $\bar \sigma(G)$ whose program does
fit into our framework.  
They throw out all the constraints of Equation~(3) which are off-diagonal, and remove the 
non-negativity constraints for the on-diagonal blocks of $X$.  
More precisely, they consider the following program:
\begin{align}
\bar \sigma(G)=&\max_{X} \ \frac{1}{2} \widehat C \bullet X \mbox{ such that} \\
& \sum_{u,w \in U} |X[(s,u),(s',w)]| \le 1 \mbox{ for all } s,s' \in S \\
& \sum_{u,w \in W} |X[(t,u),(t',w)]| \le 1 \mbox{ for all } t,t' \in T \\
& X \bullet E_{(s,u),(t,w)} \ge 0 \mbox{ for all } s \in S, t \in T, u \in U, w \in W   \\
& X \succeq 0
\end{align}

Let us see that this program fits into the framework of \thmref{main}.  The vector of matrices $\B$ 
is composed of the matrices $E_{(s,u),(t,w)}$ for $s \in S, u \in U$ and $t \in T, w \in W$.  Each of 
these matrices is block diagonal with respect to the natural partition of $\widehat C$.  Moreover, as 
$\widehat C$ is non-negative and bipartite, we can write $\widehat C = u^T \B$ for a non-negative 
$u$, namely where $u$ is given by concatenation of the entries of $C$ and $C^T$ written as a 
long vector.

The on-diagonal constraints given by Equations~(7), (8) are
not immediately seen to be of the form needed for \thmref{main} for two reasons: first, they are 
inequalities rather than equalities, and second, they have of absolute value signs.  Fortunately,
both of these problems can be easily dealt with.

It is not hard to check that \thmref{main} also works for inequality constraints 
$\A \bullet X \le b$.  The only change needed is that in the dual formulation we have the additional
constraint $y \ge \veczero$.  This condition is preserved in the product solution 
constructed in the proof of \thmref{main} as $y \otimes y \ge 0$.  

The difficulty in allowing constraints of the form $\A \bullet X \le b$ is in fact that the opposite direction 
$\alpha(\pi_1 \times \pi_2) \ge \alpha(\pi_1)\alpha(\pi_2)$ does not hold in general.  Essentially, what
can go wrong here is that $a_1, a_2 \le b$ does 
not imply $a_1 a_2 \le b^2$.  In our case, however, this does not occur as all the terms involved 
are positive and so one can show 
$\bar \sigma(G_1 \times G_2) \ge \bar \sigma(G_1) \bar \sigma(G_2)$.
 
To handle the absolute value signs we consider an equivalent formulation of 
$\bar \sigma(G)$.  We replace the condition that the sum of absolute values is at most one by 
constraints saying that the sum of every possible $\pm$ combination of values is at most one:
\begin{align*}
\bar \sigma'(G)=&\max_{X} \frac{1}{2} \widehat C \bullet X \mbox{ such that} \\
& \sum_{u,w \in U} (-1)^{x_{uw}} X[(s,u), (s',w)] \le 1 \mbox{ for all } s,s' \in S \mbox { and } 
x \in \{0,1\}^{|U|^2} \\
& \sum_{u,w \in W} (-1)^{x_{uw}} X[(t,u),(t',w)] \le 1 \mbox{ for all } t,t' \in T \mbox { and } 
x \in \{0,1\}^{|W|^2} \\
& X \bullet E_{(s,u),(t,w)} \ge 0 \mbox{ for all } s \in S, t \in T, u \in U, w \in W   \\
& X \succeq 0
\end{align*}

This program now satisfies the conditions of \thmref{main}.
It is clear that $\bar \sigma(G)=\bar \sigma'(G)$, and also that this equivalence is preserved under
product.  Thus the product theorem for $\bar \sigma(G)$ follows from \thmref{main} as well.

\section{Conclusion}
We have now developed a theory which covers all examples of semidefinite programming 
product theorems we are aware of in the literature.  Having such a theory which can be applied in 
black-box fashion should simplify the pursuit of product theorems via semidefinite programming
methods, and we hope will find future applications.  That being said, we still think there is more
work to be done to arrive at a complete understanding of semidefinite product theorems.  
In particular, we do not know the extension of item~(1) of \thmref{MS} to the case of non-negative
constraints, and it would nice to understand to what extent item~(2) of \thmref{main} can 
be relaxed. 

So far we have only considered tensor products of programs.  One could also try for more
general {\em composition} theorems: in this setting, if one has a lower bound on the complexity 
of $f:\01^n \rightarrow \01$ and $g:\01^k \rightarrow \01$, one would like to obtain a lower bound
on $(f \circ g) (\vec x)=f(g(x_1), \ldots, g(x_n))$.  What we have studied so far in looking at tensor
products corresponds to the special cases where $f$ is the PARITY or AND function, depending on
if the objective matrix is a sign matrix or a $0/1$ valued matrix.  One example of such a 
general composition theorem is known for the adversary method, a semidefinite programming 
quantity which lower bounds  quantum query complexity.  There it holds that 
$\ADV(f\circ g) \ge \ADV(f) \ADV(g)$ \cite{Amb03,HLS07}.  It would be interesting to develop a theory 
to capture these cases as well. 

\section*{Acknowledgements}
We would like to thank Mario Szegedy for many insightful conversations.  We would also like to 
thank the anonymous referees of ICALP 2008 for their helpful comments.

\newcommand{\slasha}{\discretionary{/}{/}{/}}

\end{document}